\documentclass[amsmath,amssymb,twocolumn,10pt]{revtex4}
\usepackage{mathbbold}
\usepackage{amssymb}
\usepackage{multirow}
\usepackage{color}
\usepackage{amsthm}
\def\beq{\begin{equation}}
\def\eeq{\end{equation}}
\usepackage{amsthm}
\usepackage{dcolumn}
\usepackage{bm}
\usepackage{graphicx}
\usepackage[paperwidth=18.4cm, paperheight=26cm,left=2.5cm,right=2.0 cm,top=2.0cm,bottom=2.5cm,includehead,includefoot]{geometry}

\begin{document}
\newtheorem{corollary}{Corollary}
\newtheorem{definition}[corollary]{Definition}
\newtheorem{example}[corollary]{Example}
\newtheorem{lemma}[corollary]{Lemma}
\newtheorem{proposition}[corollary]{Proposition}
\newtheorem{theorem}[corollary]{Theorem}
\newtheorem{fact}[corollary]{Fact}
\newtheorem{property}[corollary]{Property}
\newtheorem{observation}[corollary]{Observation}

\newcommand{\multicols}{
\newlength{\halfpagewidth}
\setlength{\halfpagewidth}{\linewidth}
\divide\halfpagewidth by 2
\newcommand{\leftsep}{%
\noindent\raisebox{4mm}[0ex][0ex]{%
\makebox[\halfpagewidth]{\hrulefill}\hbox{\vrule height 3pt}}%
\vspace*{-2mm}%
}
\newcommand{\rightsep}{%
\noindent\hspace*{\halfpagewidth}%
\rlap{\raisebox{-3pt}[0ex][0ex]{\hbox{\vrule height 3pt}}}%
\makebox[\halfpagewidth]{\hrulefill}%
}}

\newenvironment{remark}{\noindent \textbf{{Remark~}}}{\qed}
\newcommand{\remarkTitle}[1]{\textbf{(#1)}}
\newcommand{\proofComment}[1]{\exampleTitle{#1}}
\newcommand{\bra}[1]{\langle #1|}
\newcommand{\ket}[1]{|#1\rangle}
\newcommand{\braket}[3]{\langle #1|#2|#3\rangle}
\newcommand{\ip}[2]{\langle #1|#2\rangle}
\newcommand{\op}[2]{|#1\rangle\!\langle #2|}
\newcommand{\tr}{{\operatorname{Tr}\,}}
\newcommand{\supp}{{\operatorname{supp}\,}}
\newcommand{\Sch}{{\operatorname{Sch}}}
\newcommand{\GHZ}{{\textrm{GHZ}}}
\newcommand{\slocc}{\stackrel{\textrm{\scriptsize SLOCC}}{\longrightarrow}}
\newcommand{\locc}{\stackrel{\textrm{\scriptsize LOCC}}{\longrightarrow}}
\newcommand{\rk}{{\operatorname{rk}}}
\newcommand{\sr}{{\operatorname{srk}}}
\newcommand{\pr}{{\operatorname{pr}}}
\newcommand{\E}{{\mathcal{E}}}
\newcommand{\F}{{\mathcal{F}}}
\newcommand{\h}{\mathcal{H}}
\newcommand{\diag} {{\rm diag}}
\newcommand{\nc}{\newcommand}
\nc{\ox}{\otimes}
\title{Common Resource State for Preparing Multipartite Quantum Systems via Local Operations and Classical Communication}
\author{Cheng Guo$^{1~4}$}
\author{Eric Chitambar$^3$}
\author{Runyao Duan$^{2~4}$}
\affiliation{$^1$ Institute For Advanced Study, Tsinghua University, Beijing 100084, China}
\affiliation{$^2$ Institute for Quantum Computing, Baidu Inc., Beijing 100193, China}
\affiliation{$^3$ Department of Physics and Astronomy, Southern Illinois University, Carbondale, Illinois 62901, USA}
\affiliation{$^4$ Centre for Quantum Software and Information, Faculty of Engineering and Information Technology, University of Technology, Sydney, NSW 2007, Australia}

\begin{abstract}
Given a set of multipartite entangled states, can we find a common state to prepare them by local operations and classical communication?
Such a state, if exists, will be a common resource for the given set of states.  We completely solve this problem for bipartite pure states case by explicitly constructing a unique optimal common resource state for any given set of states. In the multipartite setting, the general problem becomes quite complicated, and we focus on finding nontrivial common resources for the whole multipartite state space of given dimensions.  We show that $\ket{GHZ_3}=1/\sqrt{3}(\ket{000}+\ket{111}+\ket{222})$ is  a nontrivial common resource for three-qubit systems.
\end{abstract}
\maketitle

The problem of transforming one entangled state to another by local operations and classical communication (LOCC) is of central importance in quantum entanglement theory. This problem has been studied extensively in last two decades, and many interesting results have been reported. Notably, Nielsen pointed out that local transformations among bipartite pure states can be completely characterized by an algebraic relation of majorization between their Schmidt coefficient vectors \cite{Nie99}. The majorization characterization can be extended to a class of multipartite pure states having Schmidt decompositions \cite{XD2007}. Unfortunately, the Schmidt decomposition for a generic multipartite pure state doesn't exist, and it is still an open problem to determine whether a general multipartite pure state can be transformed into another one by LOCC.

Despite the overall complexity of multipartite entanglement transformations, we can often find entangled states that can be transformed into any other state in the \textit{same} state space by LOCC. These kind of states are called \textit{maximally entangled states}, and they exist in spaces if and only if the dimension of one subsystem is no less than the product of dimensions of all other subsystems \cite{DS2009}.  Bell states are one such example, which can be transformed into any pure state in two-qubit systems.  As a simple corollary of the dimensionality bound, there is no maximally entangled state in three-qubit systems.  In fact, it is well-known that three-qubit states can be entangled in two different ways, one class consisting of so-called W-type states and other consisting of Greenberge--Horn--Zeilinger (GHZ)-type states \cite{DVC00}.  While one copy of $\ket{GHZ}$ cannot be stochastically transformed into $\ket{W}$, interestingly this transformation becomes possible at an asymptotic rate approaching unity if multi-copy transformation is allowed \cite{YGD2014}.

In this paper, we generalize the notion of maximally entangled states with respect to LOCC transformations. The problem we study can be best described through the following scenario. Assume that Alice and Bob are going to implement a series of quantum information tasks, each one requiring a different entangled state to perform.  However, instead of sharing a multitude of different states, they wish to share only \textit{one} type of entangled state and then transform this state into a different form as needed.  Thus, the question is:  for a given set of pure entangled states, is there a certain state which can be locally transformed into all of them by LOCC?  Below, we give a complete solution to this problem for bipartite systems. In addition, we study the case when the set of target states is the entire state space for some given dimensions.  When the
dimension of one subsystem in this target space is not smaller than the product of all the other subsystem dimensions, there always exists some (perhaps higher-dimensional) state that can obtain all states in the target space \cite{DS2009}.  However, this dimensionality condition turns out not to be necessary. Interestingly, we find a non-trivial state $\ket{GHZ_3}=\frac{1}{\sqrt{3}}(\ket{000}+\ket{111}+\ket{222})$ which can be transformed into any three-qubit pure state by LOCC.

Let us now formulate our problem more precisely. Let $S=\{\ket{\psi_1},\ket{\psi_2}, \cdots\}$ be a set of (multipartite) states, possibly infinite.
A common resource state $\ket{\psi}$ to $S$ can be transformed into any state in $S$ by LOCC. We say $\ket{\psi}$ an \textit{optimal common resource} (OCR) if for any other common resource $\ket{\phi}$ we have either $\ket{\phi}$ can be transformed into $\ket{\psi}$ by LOCC, or
$\ket{\phi}$ and $\ket{\psi}$ are not comparable under LOCC.
In general, it is a hard problem to find OCR for a set of multipartite states. For bipartite pure states, majorization characterizes the LOCC transformation between two pure state \cite{Nie99}. Given a bipartite pure state $\ket{\psi}$, $\lambda^\psi$ denotes a probability vector whose entries are in descending order of the Schmidt coefficients of $\ket{\psi}.$ For instance, if $\ket{\psi}=\frac{1}{\sqrt{2}}\ket{0}\ket{0}+\frac{1}{\sqrt{6}}\ket{1}\ket{2}+\frac{1}{\sqrt{3}}\ket{2}\ket{1},$
$\lambda^\psi=(\frac{1}{2},\frac{1}{3},\frac{1}{6}).$
If $\lambda^\psi=(x_1,\cdots,x_d)$ and $\lambda^\phi=(y_1,\cdots,y_d)$ satisfy
$i)$ $\forall k, 1 \leq k \leq d, \sum \limits_{j=1}^k x_j \leq \sum \limits_{j=1}^k y_j,$
$ii)$ $\sum \limits_{j=1}^d x_j = \sum \limits_{j=1}^d y_j,$
we say that $\lambda^{\psi}$ is majorized by $\lambda^{\phi}$ and write $\lambda^{\psi} \prec \lambda^{\phi}.$
Nielsen established the following fundamental result: A bipartite pure state $\ket{\psi}$ can be transformed to another pure state $\ket{\phi}$ by LOCC if and only if $\lambda^{\psi} \prec \lambda^{\phi}$ \cite{Nie99}.

Nielsen's result together with the properties of majorization leads us to an explicit construction of the unique OCR of a set of bipartite pure states.
\begin{theorem} \label{ocrv}
Let $S=\{\ket{\phi_i},i\in I\}$ be a set of $d\otimes d$ pure states, where $I$ is an index set (finite or infinite). Assume that the Schmidt coefficient vector of $\ket{\phi_i}$ is given by $\lambda_{\phi_i}=(x_{1}^{(i)},\cdots,x_{d}^{(i)})$.

Then if $I$ is a finite set, the OCR for $S$ always exists and is unique. The OCR state $\ket{\psi}$ is given by $\lambda^{\psi}=(y_1,\cdots,y_d),$ where

$y_k= \min_{i\in I}\sum \limits_{j=1}^k x_{j}^{(i)}-\min_{i\in I}\sum \limits_{j=1}^{k-1} x_{j}^{(i)}.$

Furthermore, if $I$ is infinite, the $\min$ sign in above equations should be replaced with $\inf$.

\end{theorem}

The proof is given in the Appendix.

Let us consider an example to demonstrate the application of the above theorem. Let a $d-$dimensional bipartite target set be $S_a=\{ \ket{\phi} | \lambda^{\phi}_1 \geq a \}$ , where $a \geq 1/d$. Then an OCR $\ket{\psi}$ for $S_a$ can be chosen as $\ket{\psi}$ such that $\lambda^{\psi}=(a,\frac{1-a}{d-1},\frac{1-a}{d-1} , \cdots ,\frac{1-a}{d-1})$. The maximal entangled state $\frac{1}{\sqrt{d}}\sum \limits_{k=0}^{d-1}\ket{kk}$ is always a common resource but usually not optimal.

We shall now move to multipartite setting, and consider the problem ``what is the common resource of the whole system''. For $n-$partite quantum system $d_1\otimes \cdots \otimes d_n$ where $d_1\geq d_2\geq \cdots\geq d_n$, the maximal entangled state exists, in the sense that all other states in the system can be obtained from the state by LOCC, if and only if $d_1\geq \prod \limits_{i=2}^n d_i$ \cite{DS2009}. For instance, the state $\frac{1}{2}(\ket{000}+\ket{101}+\ket{210}+\ket{311})$ is an OCR of tripartite $\h_4 \ox \h_2 \ox \h_2$ system. Interestingly, the OCR exists even if any sub-system's dimension is less than the product of other sub-system's dimensions.

What is the connection between $\ket{\psi}$ being an OCR for $S$ and $S$ being a possible multi-outcome image for some LOCC map performed on $\ket{\psi}$?  The latter question was first posed by Jonathan and Plenio \cite{JP2014}, and is more precisely stated as follows: If $p(i)$ is an arbitrary distribution over the $\ket{\phi_i}\in S$, when does there exist an LOCC transformation that transforms $\ket{\psi}$ into $\ket{\phi_i}$ with probability $p(i)$ for each $\ket{\psi_i}$?  If $\ket{\psi}$ is an OCR for $S$, then clearly $\ket{\psi}$ can generate the ensemble $\{\ket{\psi_i},p(i)\}_{i\in I}$;  Alice and Bob simply sample from $I$ with distribution $p(i)$ and then deterministically prepare $\ket{\psi_i}$ given $i$.  We can also see the possibility of generating $\{\ket{\psi_i},p(i)\}_{i\in I}$ directly from Theorem 1.  For any $l\in\{1,\cdots, d\}$ we have

$\sum_{k=1}^l y_k
\\ =\sum_{k=1}^l\left(\min\limits_{i\in I}\sum \limits_{j=1}^k x_{j}^{(i)}-\min\limits_{i\in I}\sum \limits_{j=1}^{k-1} x_{j}^{(i)}\right)
\\ =\min\limits_{i\in I}\sum \limits_{j=1}^l x_{j}^{(i)}\leq \sum_{i\in I}p(i)\sum \limits_{j=1}^l x_{j}^{(i)}.\notag $

Satisfying this inequality is both necessary and sufficient for transforming $\ket{\psi}$ into the ensemble $\{\ket{\psi_i},p(i)\}_{i\in I}$ \cite{JP2014}.

\begin{theorem}
$\ket{GHZ_3}=\frac{1}{\sqrt{3}}(\ket{000}+\ket{111}+\ket{222})$ is a common resource of three-qubit system.
\end{theorem}
\begin{proof}
If $\ket{\psi}$ can be transformed into $\ket{\phi}$ via LOCC, we denote this as $\ket{\psi} \locc \ket{\phi}.$

The proof details of the protocols will be given in the Appendix. Here is an outline of the proof ideas.  We divide the proof into two cases according to the target states are W-type or GHZ-type states. The case of GHZ can be further divide into two sub-cases: orthogonal GHZ and non-orthogonal GHZ.

If the target state $\ket{\phi}$ is SLOCC equivalent to a W state, $\ket{\phi}$ is local-unitary (LU) equivalent to $x_0 \ket{000}+x_1 \ket{100}+x_2 \ket{010}+x_3 \ket{001},$ where $x_k$ are all positive real numbers and $\sum \limits_{k=0}^3 x_k^2=1$ \cite{KT2010}. Then, we can do the following transforms:
$\\i)$ $\ket{GHZ_3} \locc \sqrt{x_0^2+x_1^2}\ket{000} + x_2 \ket{111}+x_3 \ket{222}$
$\\ii)$ $\sqrt{x_0^2+x_1^2}\ket{000} + x_2 \ket{111}+x_3 \ket{222} \locc \sqrt{x_0^2+x_1^2}\ket{100} + x_2 \ket{010}+x_3 \ket{001}$
$\\iii)$ $\sqrt{x_0^2+x_1^2}\ket{100} + x_2 \ket{010}+x_3 \ket{001}
\locc \ket{\phi}.$

The step $i)$ has been shown in \cite{XD2007}, and the step $iii)$ was given in \cite{KT2010}. For completeness we provided simple proofs in appendix.

If the target state $\ket{\phi}$ is SLOCC equivalent to GHZ state,
$\ket{\phi}$ is LU equivalent to $x \ket{000} + y \ket{\phi_A \phi_B \phi_C}$. When $\ip{000}{\phi_A \phi_B \phi_C}=0$, we name it orthogonal GHZ case, otherwise non-orthogonal GHZ case \cite{TGP2010}.

For the orthogonal GHZ case, $\ket{GHZ_2}=\frac{1}{\sqrt{2}}(\ket{000}+\ket{111}) \locc \ket{\phi}$ \cite{TGP2010}. Notice that, $\ket{GHZ_3} \locc \ket{GHZ_2}$ \cite{XD2007}. Hence, $\ket{GHZ_3} \locc \ket{\phi}.$

For the non-orthogonal GHZ case, suppose $\ket{\phi_A}=a_0 \ket{0} + a_1 \ket{1}$, $\ket{\phi_B}=b_0 \ket{0} + b_1 \ket{1}$ and $\ket{\phi_C}=c_0 \ket{0} + c_1 \ket{1}$, and $a_0, a_1,b_0, b_1,c_0$ and $c_1$ are all real numbers.
$\ket{\mathit{\Psi}_1}=\sqrt{|x c_0 +y a_0 b_0|^2 + |y a_1 b_0|^2}~ \ket{000} + |x c_1| \ket{111} + |y b_1| \ket{222}$ and $\ket{\mathit{\Psi}_2}=\sqrt{|x c_0 +y a_0 b_0|^2 + |y a_1 b_0|^2}~ \ket{000} + |x c_1| \ket{101} + |y b_1| \ket{210}$

$\\i)$ $\ket{GHZ_3} \locc \ket{\mathit{\Psi}_1}$
$\\ii)$ $\ket{\mathit{\Psi}_1} \locc \ket{\mathit{\Psi}_2}$
$\\iii)$ $\ket{\mathit{\Psi}_2} \locc x \ket{00 \phi_C} + y \ket{ \phi_A \phi_B 0}$
$\\iv)$ $x \ket{00 \phi_C} + y \ket{ \phi_A \phi_B 0} \locc x \ket{000} + y \ket{ \phi_A \phi_B  \phi_C}$

Again the step $i)$ is from \cite{XD2007}.
The step $iii)$ is very technical, and is given in appendix.
The step $iv)$: Charlie takes a local unitary operation $\op{\phi_C}{0}+(c_1\ket{0} - c_0 \ket{1}) \langle 1|.$
\end{proof}

We further check the $3 \ox 3 \ox 2$ system and get the following result.

\begin{theorem}
There is no common resource of three-qubit system in $\h^3 \ox \h^2 \ox \h^2 $ systems.
\end{theorem}

The proof is quite technical. The main idea is: there are two SLOCC equivalence classes in $\h^3 \ox \h^2 \ox \h^2$ system. One is SLOCC equivalent to $\ket{000}+\ket{101}+\ket{210}$ ($A$)
, and another is with rank $\ket{000}+\ket{111}+\ket{210}+\ket{201}$.
Any state with the first type in that system cannot be transformed into $\ket{GHZ_2}$ via LOCC. Also, no state with the second type in that system can be transformed into any $W-$kind of states via LOCC.We will show the details of this proof in Appendix.

In conclusion, we introduce a notion of optimal common resource for a set of entangled states, and explicitly construct it for any bipartite pure state set. We also show that $\ket{GHZ_3}$ state is a nontrivial common resource for three-qubit system, and conjecture its optimality.
We are still not sure about whether there exists a common resource of three-qubit system is $\h^3 \ox \h^3 \ox \h^2$ systems. We hope this problem will stimulate further research interest in entanglement transformation theory.

CG and RD were supported in part by the Australian Research Council (Grants No. DP120103776 and FT120100449) and the National Natural Science Foundation of China (Grant No. 61179030).
CG and EC were supported in part by the National Science Foundation (NSF) Early CAREER Award No. 1352326.

\newpage
\begin{center}
\Large Appendix
\end{center}

\section{proof of theorem \ref{ocrv}}

According to Nielsen's majorization criterion for entanglement transformation,  Theorem \ref{ocrv} is essentially due to the following lemma:
\begin{lemma}
Suppose that $X^{(k)}=(x_{1}^{(k)},x_{2}^{(k)},\cdots,x_{d}^{(k)}$ are a set of $d$--dimensional vectors where
$x_{1}^{(k)}\!\geq x_{2}^{(k)}\! \geq\! \cdots\! \geq\!x_{d}^{(k)}.$
There always exists an optimal vector $Y=(y_1,\cdots,y_d)$ such that $Y\prec X^{(k)}$ for any $k$. Furthermore,  $Y=(y_1,\cdots, y_d)$ can be chosen as
$y_k= \min(\sum \limits_{j=1}^k x_{j}^{(1)},\sum \limits_{j=1}^k x_{j}^{(2)},\cdots,\sum \limits_{j=1}^k x_{j}^{(n)})-\min(\sum \limits_{j=1}^{k-1} x_{j}^{(1)},\sum \limits_{j=1}^{k-1} x_{j}^{(2)},\cdots,\sum \limits_{j=1}^{k-1} x_{j}^{(n)}).$ Furthermore, if n is infinite, $y_k= \inf(\sum \limits_{j=1}^k x_{j}^{(h)}|h=1,2 \cdots)-\inf(\sum \limits_{j=1}^{k-1} x_{j}^{(h)}|h=1,2 \cdots).$
\end{lemma}


\begin{proof}
Let us first consider the case that $n$ is finite. We will complete the proof via the following steps:

Step 1: $y_k \geq y_{k+1}.$

Suppose $\min(\sum \limits_{j=1}^k x_{j}^{(1)},~\sum \limits_{j=1}^k x_{j}^{(2)},\cdots,~\sum \limits_{j=1}^k x_{j}^{(n)})$ is just $\sum \limits_{j=1}^k x_{j}^{(t)}$.

$y_{k+1}=\min(\sum \limits_{j=1}^{k+1} x_{j}^{(1)},~\sum \limits_{j=1}^{k+1} x_{j}^{(2)},\cdots,~\sum \limits_{j=1}^{k+1} x_{j}^{(n)})-\min(\sum \limits_{j=1}^{k} x_{j}^{(1)},~\sum \limits_{j=1}^{k} x_{j}^{(2)},\cdots,~\sum \limits_{j=1}^{k} x_{j}^{(n)})
\\ \leq x_{j+1}^{(t)} \leq x_{j}^{(t)}
\\ \leq \min(\sum \limits_{j=1}^k x_{j}^{(1)},~\sum \limits_{j=1}^k x_{j}^{(2)},\cdots,~\sum \limits_{j=1}^k x_{j}^{(n)})-\min(\sum \limits_{j=1}^{k-1} x_{j}^{(1)},~\sum \limits_{j=1}^{k-1} x_{j}^{(2)},\cdots,~\sum \limits_{j=1}^{k-1} x_{j}^{(n)}) =y_k.
$

Step 2: $\forall m,~Y \prec X^{(m)}.$

It is obvious to see that $\sum \limits_{j=1}^k y_{j} \leq \sum \limits_{j=1}^k x_{j}^{(m)}.$

Step 3: If $\forall k,~Z \prec X^{(k)},$ then $Z\prec Y.$

Otherwise, if $\exists k_0, \sum \limits_{j=1}^{k_0} z_{j} \textgreater \sum \limits_{j=1}^{k_0} y_{j}.$

$\sum \limits_{j=1}^{k} y_{j} = \min(\sum \limits_{j=1}^k x_{j}^{(1)},
\cdots,\sum \limits_{j=1}^k x_{j}^{(n)}),$
so $\forall k, \exists f(k), \sum \limits_{j=1}^{k} y_{j} = \sum \limits_{j=1}^k x_{j}^{(f(k))}.$

$\sum \limits_{j=1}^{k_0} z_{j} \textgreater \sum \limits_{j=1}^{k_0} y_{j} = \sum \limits_{j=1}^{k_0} x_{j}^{(f(k_0))}.$ This is a contradiction against with $Z \prec X^{(f(k_0))}.$

If $n$ is infinite, we modify our proof as follows: (notice that:)

$\forall k,m, \sum \limits_{j=1}^{k} y_{j} = \inf(\sum \limits_{j=1}^k x_{j}^{(h)}) \leq \sum \limits_{j=1}^k x_{j}^{(m)}.$

Step 1: $y_k \geq y_{k+1}.$

Suppose $\inf(\sum \limits_{j=1}^{k} x_{j}^{(h)}) = \lim \limits_{h\rightarrow \infty} \sum \limits_{j=1}^{k} x_{j}^{(g(h))},$ where $\sum \limits_{j=1}^{k}x_{j}^{(g(h))}$ is in descending order. We can also find a sub--sequence $(f(h)) \subset $ sequence$(g(h)),$ such that $\inf(\sum \limits_{j=1}^{k} x_{j}^{(h)}) = \lim \limits_{h\rightarrow \infty} \sum \limits_{j=1}^{k} x_{j}^{(f(h))}$, also $\lim \limits_{h\rightarrow \infty} x_{k}^{(f(h))}$
and $\lim \limits_{h\rightarrow \infty} x_{k+1}^{(f(h))}$ exist.
This is because any infinite bounded sequence must have a monotonic convergent sub--sequence with a limit.

$y_{k+1}=\inf(\sum \limits_{j=1}^{k+1} x_{j}^{(h)})-
\inf(\sum \limits_{j=1}^{k} x_{j}^{(h)}) = \inf(\sum \limits_{j=1}^{k+1} x_{j}^{(h)})- \lim \limits_{h\rightarrow \infty} \sum \limits_{j=1}^{k} x_{j}^{f(h)}
\\ \leq \lim \limits_{h\rightarrow \infty} \sum \limits_{j=1}^{k+1} x_{j}^{f(h)} - \lim \limits_{h\rightarrow \infty} \sum \limits_{j=1}^{k} x_{j}^{f(h)}
\\ \leq \lim \limits_{h\rightarrow \infty} x_{k+1}^{f(h)} \leq \lim \limits_{h\rightarrow \infty} x_{k}^{f(h)}
\\ \leq  \lim \limits_{h\rightarrow \infty} \sum \limits_{j=1}^{k} x_{j}^{f(h)}-\inf(\sum \limits_{j=1}^{k-1} x_{j}^{(h)})
=\inf(\sum \limits_{j=1}^{k} x_{j}^{(h)})-
\inf(\sum \limits_{j=1}^{k-1} x_{j}^{(h)}) =y_k.
$

Step 2: $\forall m,~Y \prec X^{(m)}.$

Step 3: If $\forall k,~Z \prec X^{(k)},$ then $Z\prec Y.$

Otherwise, if $\exists k_0, \sum \limits_{j=1}^{k_0} z_{j} \textgreater \sum \limits_{j=1}^{k_0} y_{j}.$

$\sum \limits_{j=1}^{k_0} z_{j} \textgreater \sum \limits_{j=1}^{k_0} y_{j} = \inf(\sum \limits_{j=1}^k x_{j}^{(h)}).$ This is a contradiction against with $\forall k,~Z \prec X^{(k)}.$

\end{proof}

\section{LOCC transformation protocols from $GHZ_3$ to any three-qubit pure state}

This section is the details of the protocol transforming $\ket{GHZ_3}$ into any three-qubit pure state.
The proof is divided into 2 parts: W and GHZ. GHZ part have two cases: orthogonal GHZ and non-orthogonal GHZ case.

The following lemma has been shown in \cite{XD2007}. This  also is our step $i)$ in all the cases.
\begin{lemma} \label{schmidt}
$\forall z_0,z_1,z_2, \sum \limits_{k=0}^2 |z_k|^2=1$,
$$\ket{GHZ_3} \locc  z_0\ket{000}+ z_1\ket{111}+ z_2\ket{222}.$$
\end{lemma}
\begin{proof}
Alice takes the following measurement and sends Bob and Charlie the result,
$\\ \{M_1= z_0 \op{0}{0} + z_1 \op{1}{1} + z_2 \op{2}{2},
\\ ~~M_2= z_0\op{0}{1} + z_1 \op{1}{2} + z_2 \op{2}{0},
\\ ~~M_3= z_0 \op{0}{2} + z_1 \op{1}{0} + z_2 \op{2}{1} \}.$

The, Bob and Charlie make some unitary operations based on Alice's measurement outcome, which transform the state to $z_0\ket{000}+ z_1\ket{111}+ z_2\ket{222}.$
\end{proof}

\subsection{Protocol of entanglement transformation from $GHZ_3$ to $W$ type states}
If the target state $\ket{\phi}$ is stochastic local operations and classical communication (SLOCC) equivalent to W state, $\ket{\phi}$ can be written as $x_0 \ket{000}+x_1 \ket{100}+x_2 \ket{010}+x_3 \ket{001},$ where $x_k$ are all positive real numbers and $\sum \limits_{k=0}^3 x_k^2=1$ \cite{KT2010}.
$\\i)$ $\ket{GHZ_3} \locc \sqrt{x_0^2+x_1^2}\ket{000} + x_2 \ket{111}+x_3 \ket{222}$
$\\ii)$ $\sqrt{x_0^2+x_1^2}\ket{000} + x_2 \ket{111}+x_3 \ket{222} \locc \sqrt{x_0^2+x_1^2}\ket{100} + x_2 \ket{010}+x_3 \ket{001}$
$\\iii)$ $\sqrt{x_0^2+x_1^2}\ket{100} + x_2 \ket{010}+x_3 \ket{001} \locc \ket{\phi}.$

Step $i)$ is lemma \ref{schmidt} which $z_0=\sqrt{x_0^2+x_1^2},~z_1=x_2$ and $z_2=x_3$.

Step $ii):$ $\sqrt{x_0^2+x_1^2}\ket{000} + x_2 \ket{111}+x_3 \ket{222} \locc \sqrt{x_0^2+x_1^2}\ket{100} + x_2 \ket{010}+x_3 \ket{001}$

Alice takes the measurement:
$\{M_1=(\op{1}{0} + \op{0}{1}+\op{0}{2})/ \sqrt{2}, ~M_2=(\op{1}{0}+\op{0}{2}-\op{0}{1})/ \sqrt{2} \}.$

Bob takes the measurement:
$\{M_1=(\op{1}{1} + \op{0}{0}+\op{0}{2})/ \sqrt{2}, ~M_2=(\op{1}{1}+\op{0}{0}-\op{0}{2})/ \sqrt{2} \}.$

Charlie takes the measurement$\\iii)$ $\sqrt{x_0^2+x_1^2}\ket{100} + x_2 \ket{010}+x_3 \ket{001} \locc \ket{\phi}.$
$\{M_1=(\op{1}{2} + \op{0}{1}+\op{0}{0})/ \sqrt{2}, ~M_2=(\op{1}{2}+\op{0}{1}-\op{0}{0})/ \sqrt{2} \}.$

Alice transmits her result to Bob.
Bob transmits his result to Charlie.
Charlie transmits its result to Alice.

If the result is $2$, he or she should take a $Z-$operation, $Z=\op{0}{0}-\op{1}{1}.$

Now, the state is $\sqrt{x_0^2+x_1^2}\ket{100} + x_2 \ket{010}+x_3 \ket{001}.$

Step $iii):$ $\sqrt{x_0^2+x_1^2}\ket{100} + x_2 \ket{010}+x_3 \ket{001} \locc \ket{\phi}.$

Alice takes the measurement:
$\\ \{M_1=\frac{1}{\sqrt{2}}(\op{0}{0} + \frac{x_1}{\sqrt{x_0^2+x_1^2}} \op{1}{1} +  \frac{x_0}{\sqrt{x_0^2+x_1^2}} \op{0}{1}),
\\ ~~ M_2=\frac{1}{\sqrt{2}} (\op{0}{0} + \frac{x_1}{\sqrt{x_0^2+x_1^2}} \op{1}{1} -  \frac{x_0}{\sqrt{x_0^2+x_1^2}} \op{0}{1}) \}.$

Now, the state is $(\pm x_0 \ket{000} + x_1 \ket{100} + x_2 \ket{010} + x_3 \ket{001}).$ If Alice's result is $1$, we already get the target. If it is $2$, Alice transmits $2$ to Bob and Charlie. Bob and Charlie make unotary operation $Z=\op{0}{0}- \op{1}{1}$ on their own part. Finally, Alice makes an unitary operation $ -Z=-\op{0}{0} + \op{1}{1}$. Now, the state is $x_0\ket{000}+ x_1 \ket{100} + x_2 \ket{010} + x_3 \ket{001}.$

\subsection{Protocol of entanglement transformation from $GHZ_3$ to $GHZ$ type states}
Without loss of generality, a $GHZ$ type pure states can be denoted as $x\ket{000}+y\ket{\phi_A \phi_B \phi_C}$, where $a_0= \ip{0}{\phi_A}, a_1= \ip{1}{\phi_A}, b_0= \ip{0}{\phi_B}, b_0= \ip{0}{\phi_B}, c_0= \ip{0}{\phi_C}, c_0= \ip{0}{\phi_C}.$ $a_0, a_1, b_0, b_1, c_0$ and $c_1$ are all real numbers.

The LOCC protocols from $GHZ_3$ to orthogonal GHZ state (suppose $c_0=0$, $\ket{\phi_C}=\ket{1}$. Thus, $|x|^2+|y|^2=1$):
$\\i)$ $\ket{GHZ_3} \locc \ket{GHZ_2}.$
$\\ii)$ $\ket{GHZ_2} \locc \ket{\phi}.$

Step $i)$ is for lemma \ref{schmidt} which $z_0=z_1=1/\sqrt{2}$ and $z_2=0.$ This is a protocol $\ket{GHZ_3} \locc \ket{GHZ_2}.$

Step $ii)$:
$\ket{GHZ_2} \locc x \ket{000}+y \ket{111} \locc x \ket{000}+y \ket{\phi_A 1 1}  \locc x\ket{000}+y\ket{\phi_A \phi_B 1} = \ket{\phi}.$

For $\ket{GHZ_2} \locc (x \ket{000}+y \ket{111})$:
Alice takes the measurement:
$\\ \{M_1=x \op{0}{0} + y\op{1}{1}, ~ M_2=x \op{1}{1} + y\op{0}{0} \}.\\$ If output is ``1'', finish. If output is ``2'', Alice take an $X-$operation, $X=\op{1}{0}+\op{0}{1}.$

The protocol for $x \ket{000}+y \ket{111} \locc x \ket{000}+y \ket{\phi_A 1 1}$ is as following.
Alice takes the measurement:
$\\ \{M_1=\frac{\sqrt{2}}{2}(\op{0}{0} + \op{\phi_A}{1}), ~ M_2=\frac{\sqrt{2}}{2}(\op{0}{0} - \op{\phi_A}{1})\}.\\$ If output is ``1'', finish. If output is ``2'', Alice transmits the result to Charlie and Charlie takes an $Z-$operation, $Z=\op{0}{0}-\op{1}{1}.$

The protocol for $x \ket{000}+y \ket{\phi_A 11} \locc x \ket{000}+y \ket{\phi_A \phi_B 1}$ is similar.
Bob takes the measurement:
$\\ \{M_1=\frac{\sqrt{2}}{2}(\op{0}{0} + \op{\phi_B}{1}), ~ M_2=\frac{\sqrt{2}}{2}(\op{0}{0} - \op{\phi_B}{1})\}.\\$ If output is ``1'', finish. If output is ``2'', Bob transmits the result to Charlie and Charlie takes an $Z-$operation, $Z=\op{0}{0}-\op{1}{1}.$

The LOCC protocols from $GHZ_3$ to non-orthogonal GHZ state:
$\\ i)$ $\ket{GHZ_3} \locc \ket{\mathit{\Psi}_1}$
$\\ ii)$ $\ket{\mathit{\Psi}_1} \locc \ket{\mathit{\Psi}_2}$
$\\ iii)$ $\ket{\mathit{\Psi}_2} \locc x \ket{00 \phi_C} + y \ket{ \phi_A \phi_B 0}$
$\\ iv)$ $x \ket{00 \phi_C} + y \ket{ \phi_A \phi_B 0} \locc x \ket{000} + y \ket{ \phi_A \phi_B  \phi_C}$

where $\ket{\mathit{\Psi}_1}=\sqrt{|x c_0 +y a_0 b_0|^2 + |y a_1 b_0|^2}~ \ket{000} + |x c_1| \ket{111} + |y b_1| \ket{222}$ and $\\ \ket{\mathit{\Psi}_2}=\sqrt{|x c_0 +y a_0 b_0|^2 + |y a_1 b_0|^2}~ \ket{000} + |x c_1| \ket{101} + |y b_1| \ket{210}.$

Step $i)$ is for lemma \ref{schmidt} which $z_0=\sqrt{|x c_0 +y a_0 b_0|^2 + |y a_1 b_0|^2},~z_1= |x c_1|$ and $z_2=|y b_1|$.

Step $ii)$ $\ket{\mathit{\Psi}_1} \locc \ket{\mathit{\Psi}_2}$:
\\ Bob takes the measurement:
$\\ \{M_1=\frac{\sqrt{2}}{2}(\op{1}{2}+\op{0}{0} + \op{0}{1}),M_2=\frac{\sqrt{2}}{2}(\op{1}{2}+\op{0}{0} - \op{0}{1}) \}.\\$
Charlie takes the measurement:
$\\ \{M_1=\frac{\sqrt{2}}{2} (\op{1}{1} + \op{0}{2}+\op{0}{0}) , M_2=\frac{\sqrt{2}}{2} (\op{1}{1} - \op{0}{2}+\op{0}{0})  \}.$

They transmit their results ($1$ or $2$) to Alice. Suppose Bob gets $\beta$ and Charlie gets $\alpha.$
Alice takes the unitary operation $M = \op{0}{0} - (-1)^{\beta} \op{1}{1} - (-1)^{\alpha} \op{2}{2} .$

Now, the state is $\ket{\mathit{\Psi}_2}.$

Step $iii)$: Alice does the measurement
$\{
\\  M_1 = (~\frac{(x c_0 + y a_0 b_0) \ket{0} + y a_1 b_0 \ket{1} }{\sqrt{|x c_0 +y a_0 b_0|^2 + |y a_1 b_0|^2}} \langle 0 |
+ \frac{x c_1}{|x c_1|} \op{0}{1}
+ \frac{y b_1}{|y b_1|} ( a_0 \ket{0} + a_1 \ket{1}) \langle 2 |~)/2,
\\  M_2 = (~\frac{(x c_0 + y a_0 b_0) \ket{0} + y a_1 b_0 \ket{1} }{\sqrt{|x c_0 +y a_0 b_0|^2 + |y a_1 b_0|^2}} \langle 0 |
+ \frac{-x c_1}{|x c_1|} \op{0}{1}
+ \frac{ y b_1 }{ |y b_1|} ( a_0 \ket{0} + a_1 \ket{1}) \langle 2 |~)/2,
\\  M_3 = (~\frac{(x c_0 + y a_0 b_0) \ket{0} + y a_1 b_0 \ket{1} }{\sqrt{|x c_0 +y a_0 b_0|^2 + |y a_1 b_0|^2}} \langle 0 | + \frac{x c_1}{|x c_1|} \op{0}{1}+ \frac{- y b_1 }{ |y b_1|} ( a_0 \ket{0} + a_1 \ket{1}) \langle 2 |~)/2,\\  M_4 = (~\frac{(x c_0 + y a_0 b_0) \ket{0} + y a_1 b_0 \ket{1} }{\sqrt{|x c_0 +y a_0 b_0|^2 + |y a_1 b_0|^2}} \langle 0 |
+ \frac{-x c_1}{|x c_1|} \op{0}{1}+ \frac{- y b_1 }{ |y b_1|} ( a_0 \ket{0} + a_1 \ket{1}) \langle 2 |~)/2\}$

The resulting states are all LOCC-equivalent to $ ((x c_0 + y a_0 b_0) \ket{0} + y a_1 b_0 \ket{1} )\ket{00} + x c_1 \ket{001} + y b_1 ( a_0 \ket{0} + a_1 \ket{1}) \ket{10} = x \ket{00 \phi_C} + y \ket{ \phi_A \phi_B 0}.$ It is LOCC-equivalent to $x \ket{000} + y \ket{ \phi_A \phi_B  \phi_C},$ which can be done in next step.

Step $iv)$: Charlie takes a local unitary operation $\op{\phi_C}{0}+(c_1\ket{0} - c_0 \ket{1}) \langle 1|.$

\subsection {The common resource of three-qubit system cannot be a $\h^3 \ox \h^2 \ox \h^2$ pure state}

We try to prove $\ket{GHZ_3}$ is an optimal common resource. We still need more time for further research.
This section shows that there is no common resource of three-qubit system in $\h^3 \ox \h^2 \ox \h^2 $ systems.

There are two SLOCC equivalent classes in $\h^3 \ox \h^2 \ox \h^2$ system \cite{CMS10}.
They are $\ket{000}+\ket{101}+\ket{210}$ and $\ket{000}+\ket{111}+\ket{201}+\ket{210}$ \cite{CMS10}.



Let
\begin{equation}
\ket{\Psi}_{ABC}=\sqrt{\lambda_0}\ket{0}\ket{\Phi_{00}}+\sqrt{\lambda_1}\ket{1}\ket{\Phi_{01}}+\sqrt{\lambda_2}\ket{2}\ket{\Phi_{10}}
\end{equation}
and
\begin{equation}
\ket{\varphi_\lambda}_{ABC}=\frac{1}{\sqrt{2}}\left(\ket{1}_A\otimes\ket{00}_{BC}+\ket{0}_A\otimes\ket{\tau_\lambda}_{BC}\right),
\end{equation}
where
\begin{equation}
\ket{\tau_\lambda}=\sqrt{1-\lambda}\ket{10}+\sqrt{\lambda}\ket{01}.
\end{equation}

\begin{theorem}
The transformation $\ket{\Psi}_{ABC}\to\ket{\varphi_\lambda}_{ABC}$ is possible by finite-round LOCC only if $\lambda=1/2$; i.e. only if $\ket{\varphi}_{ABC}$ is symmetric w.r.t. Bob and Charlie.
\end{theorem}

\begin{proof}
The proof relies on the theory of matrix pencils and its connection to $\h^3 \ox \h^2 \ox \h^2$ entanglement, as describe in Ref. \cite{CMS10}.  Recall that the set of $\h^3 \ox \h^2 \ox \h^2$ genuinely entangled states can be partitioned into two equivalence classes characterized by whether or not the matrix pencil of a given state has an elementary divisor.  It can be seen that the matrix pencil associated with the state $\ket{\Psi}_{ABC}$ has no elementary divisor and therefore it belongs to the SLOCC class (ABC-4).  We will argue that the LOCC transformation  $\ket{\Psi}_{ABC}\to\ket{\varphi_\lambda}_{ABC}$ requires passing through a state using \textit{invertible} local operations that belongs to either the SLOCC class (ABC-3) or (ABC-2), which is impossible.

Consider any finite-round LOCC protocol performed on $\ket{\psi}$.  Since $\ket{\varphi_\lambda}$ is a W-class state, it follows that Alice must be the last party to perform a non-trivial measurement along each branch of the protocol.  For one of the branches, let $\ket{\psi'}$ be the final state before Alice performs her last non-trivial measurement.  At this point, all Kraus operators performed along this branch are invertible and $\ket{\psi'}$ is still a $\h^3 \ox \h^2 \ox \h^2$ state.

We can always decompose Alice's final measurement into finite-depth tree generated by binary-outcome measurements.  Let $\ket{\psi''}$ denote one of the $\h^3 \ox \h^2 \ox \h^2$ states obtained immediately before Alice performs her final measurement along one branch of this binary-outcome tree.  Thus, Alice measures $\{M_0,M_1\}$ on $\ket{\psi''}$ and both outcomes must be LU equivalent to $\ket{\varphi_\lambda}$.  Let $\{\ket{0},\ket{1},\ket{2}\}$ be a basis in which $M_0^\dagger M_0$ is diagonal.  The condition $M_0^\dagger M_0+M_1^\dagger M_1= {I} $ implies that $M_1^\dagger M_1$ is also diagonal in this basis, and therefore we have the general forms

\begin{align}
M_0^\dagger M_0=\begin{pmatrix} \alpha&0&0\\0&1&0\\0&0&0\end{pmatrix},
\\
M_1^\dagger M_1=\begin{pmatrix} 1-\alpha&0&0\\0&0&0\\0&0&1\end{pmatrix}.
\end{align}
The polar decomposition of $M_i$ is given by $M_i=W_i\sqrt{M_i^\dagger M_i}$ for unitary $W_i$.  We next expand
$
\ket{\psi''}=c_0\ket{0}_A\ket{\chi_0}_{BC}+c_1\ket{1}_A\ket{\chi_1}_{BC}+c_2\ket{2}_A\ket{\chi_2}_{BC}.
$

Let $U_i$ and $V_i$ be the local unitaries of Bob and Charlie performed after outcome $i$ of Alice's measurement.  Since $M_0\otimes U_0\otimes V_0\ket{\psi''}\approx\ket{\varphi_\lambda}$, we must have that either (i) $U_0\otimes V_0\ket{\chi_0}=e^{i\theta_0}\ket{00}$ and $U_0\otimes V_0\ket{\chi_1}=e^{i\phi_0}\ket{\tau_\lambda}$, or (ii) $U_0\otimes V_0\ket{\chi_0}=e^{i\phi_0}\ket{\tau_\lambda}$ and $U_0\otimes V_0\ket{\chi_1}=e^{i\theta_0}\ket{00}$.

Consider first case (i).  Let $\ket{\chi_0}=\ket{\alpha\beta}$ so that $U_0\otimes V_0\ket{\alpha\beta}=e^{i\theta_0}\ket{00}$.  The other measurement outcome is $M_1\otimes U_1\otimes V_1\ket{\psi''}\approx\ket{\varphi_\lambda}$, and so we must also have $U_1\otimes V_1\ket{\alpha\beta}=e^{i\theta_1}\ket{00}$ and $U_1\otimes V_1\ket{\chi_2}=e^{i\phi_1}\ket{\tau_\lambda}$.  It  then follows that there exists a local unitary $U\otimes V$ such that $U_i\otimes V_i=(Z(\eta_i,\zeta_i)\otimes Z(\mu_i,\nu_i))(U\otimes V) $, where $Z(a,b)=\left(\begin{smallmatrix} e^{ia}&0\\0&e^{ib}\end{smallmatrix}\right)$.  Thus we have
$
\label{Eq:rank3-i}
\ket{\psi''}=c_0\ket{0}\ket{\alpha\beta}
+{I}_A\otimes U^\dagger \otimes V^\dagger 
\bigg[c_1e^{i\phi_0}I_A\otimes Z^\dagger(\eta_0,\zeta_0)\otimes Z^\dagger(\mu_0,\nu_0)\ket{1}\ket{\tau_\lambda}
\\+c_2e^{i\phi_1}I_A\otimes Z^\dagger(\eta_1,\zeta_1)\otimes Z^\dagger(\mu_1,\nu_1)\ket{2}\ket{\tau_\lambda}\bigg].
$
It is not difficult to see that state is SLOCC equivalent to the state
\[\ket{\psi'''}=\ket{000}+\ket{101}+\ket{201}.\]
One can compute that the matrix pencil associated with this state has, in fact, one elementary divisor.  Therefore, $\ket{\Psi}_{ABC}\to\ket{\psi'''}$ is impossible by LOCC.

Now consider case (ii).  We have the two conditions $U_0\otimes V_0\ket{\chi_0}=e^{i\phi_0}\ket{\tau_\lambda}$ and $U_1\otimes V_1\ket{\chi_0}=e^{i\phi_1}\ket{\tau_\lambda}$.  The key observation is that if $\lambda\not=1/2$, then the respective Schmidt bases of $\ket{\chi_0}$ and $\ket{\tau_\lambda}$ are unique, and so the actions of $U_i$ and $V_i$ are fixed up to phases.  That is, we can again write $U_i\otimes V_i=(Z(\eta_i,\zeta_i)\otimes Z(\mu_i,\nu_i))(U\otimes V) $, and $\ket{\psi''}$ will have a decomposition:
$
\label{Eq:rank3-ii}
\ket{\psi''}=c_0\ket{0}\ket{\chi_0}+I_A\otimes U^\dagger \otimes V^\dagger
\bigg[c_1e^{i\phi_0}I_A\otimes Z^\dagger(\eta_0,\zeta_0)\otimes Z^\dagger(\mu_0,\nu_0)\ket{1}\ket{00}
\\
+c_2e^{i\phi_1}I_A\otimes Z^\dagger(\eta_1,\zeta_1)\otimes Z^\dagger(\mu_1,\nu_1)\ket{2}\ket{00}\bigg].
$
It can be seen that this state is SLOCC equivalent to $\ket{0}\ket{\tau_\lambda}+\ket{1}\ket{00}$.  However, this is a $\h^2\otimes \h^2\otimes \h^2$ state which contradicts the assumption that all parties performed invertible operations up to this point in the protocol.
\end{proof}

This means any $\h^3 \ox \h^2 \ox \h^2$ state is not a common resource for three-qubit systems.\end{document}